\documentclass[aps,prl,twocolumn,superscriptaddress,floatfix]{revtex4-1} 

\usepackage[english]{babel}
\usepackage{times}

\usepackage{mathbbol}
\usepackage{amsthm}
\usepackage{mathtools}

\usepackage{paralist}

\usepackage{color}

\usepackage{amsmath}

\newtheorem{theorem}{Theorem}
\newtheorem{implication}{Implication}
\newtheorem{lemma}[theorem]{Lemma}
\renewcommand{\Re}{\mathrm{Re}}

\newcommand{\kw}[1]{\frac{1}{#1}}
\newcommand{\tkw}[1]{\tfrac{1}{#1}}
\newcommand{\mc}[1]{\mathcal{#1}}
\newcommand{\ad}{^\dagger}
\newcommand{\e}{\mathrm{e}}
\renewcommand{\i}{\mathrm{i}}
\newcommand{\1}{\mathbb{1}}
\newcommand{\del}{\partial}
\newcommand{\rmd}{\mathrm{d}}
\newcommand{\ddt}{\frac{\rmd}{\rmd t}}
\newcommand{\id}{\mathrm{id}}

\newcommand{\landauO}{\mathrm{O}}
\newcommand{\mcK}{\mc{K}}
\newcommand{\mcL}{\mathcal{L}}
\newcommand{\tr}{\mathrm{tr}}

\newcommand{\argdot}{ \cdot }
\newcommand{\Norm}[1]{\left\Vert #1 \right\Vert_{1\rightarrow 1}}

\newcommand{\nNorm}[1]{\lVert #1\rVert_{1\rightarrow 1}}
\newcommand{\BNorm}[1]{\Bigl\Vert #1 \Bigr\Vert_{1\rightarrow 1}}
\newcommand{\norm}[1]{\left\Vert #1 \right\Vert}
\newcommand{\nnorm}[1]{\lVert #1\rVert}
\newcommand{\av}{\mathrm{av}}
\newcommand{\dist}{\mathrm{dist}}
\newcommand{\mcD}{\mc{D}} %Dissipator

\newcommand{\HS}{\mathrm{HS}}
\newcommand{\SK}{\mathrm{SK}}

\newcommand{\ket}[1]{\left.\left|{#1}\right.\right\rangle}

\newcommand{\bra}[1]{\left.\left\langle{#1}\right.\right|}

\newcommand{\braket}[2]{\left\langle #1 \middle| #2 \right\rangle}

\newcommand{\ketbra}[2]{\ket{#1} \!\! \bra{#2}}

\newcommand{\constc}{2a^2+8a^3 d^k +16 a^4 d^{2k}}
\newcommand{\constb}{2 a^2(2+4 d^k)}

\newcommand{\T}[3][]{T_{\mc{#1}}(#2,#3)}
\newcommand{\iT}[3][]{T_{\mc{#1}}^-(#2,#3)}

\newcommand{\odeiT}{\eqref{odeiT} }
\newcommand{\odeT}{\eqref{odeT} }
\newcommand{\fulltrotterthm}{\ref{fulltrotterthm} }
\newcommand{\usualtriangleinequalitytrick}{\eqref{usualtriangleinequalitytrick} }
\newcommand{\numbertsteps}{\eqref{numbertsteps} }

\newcommand{\seealsoepaps}{(see the appendix)}
\newcommand{\theApndx}{the Appendix}

\newcommand{\potsdam}{Institute for Physics and Astronomy, University of Potsdam, 14476 Potsdam, Germany}
\newcommand{\fu}{Dahlem Center for Complex Quantum Systems, Freie Universit{\"a}t Berlin, 14195 Berlin, Germany}
\newcommand{\copenhagen}{Niels Bohr Institute, University of Copenhagen, 2100 Copenhagen, Denmark}

\begin{document}
\title{A dissipative quantum Church-Turing theorem}
 
\author{M.\ Kliesch}
\affiliation{\fu}
\affiliation{\potsdam}

\author{T.\ Barthel}
\affiliation{\fu}
\affiliation{\potsdam}

\author{C.\ Gogolin}
\affiliation{\fu}
\affiliation{\potsdam}

\author{M.\ Kastoryano}
\affiliation{\copenhagen}

\author{J.\ Eisert}
\affiliation{\fu}
\affiliation{\potsdam}

\begin{abstract}
We show that the time evolution of an open quantum system, described by a possibly time dependent Liouvillian, can be simulated by a unitary quantum circuit of a size scaling polynomially in the simulation time and the size of the system. 
An immediate consequence is that dissipative quantum computing is no more powerful than the unitary circuit model.
Our result can be seen as a dissipative Church-Turing theorem, since it implies that under natural assumptions, such as weak coupling to an environment, the dynamics of an open quantum system can be simulated efficiently on a quantum computer.
Formally, we introduce a Trotter decomposition for Liouvillian dynamics and give explicit error bounds. This constitutes a practical tool for numerical simulations, e.g., using matrix-product operators.
We also demonstrate that most quantum states cannot be prepared efficiently.
\end{abstract}

\pacs{ 
03.67.Ac, % Quantum algorithms and protocols -  quantum information
02.60.Cb% Numerical approximation and analysis - Numerical simulation; solution of equations
03.65.Yz, % Decoherence; open systems; quantum statistical methods
89.70.Eg, % Information theory - computational complexity
}

\maketitle

One of the cornerstones of theoretical computer science is the Church-Turing thesis \cite{CT,Vazirani}.
In its strong formulation it can be captured in the 
following way \cite{KayLafMos07,Davis}: 
``A probabilistic Turing machine can efficiently simulate any realistic model of computation.''
As such, it reduces any physical process -- that can intuitively be thought of as a computational task in a wider sense -- to what an elementary standard computer can do. 
Needless to say, in its strong formulation, the Church-Turing thesis is challenged by the very idea of a quantum computer, and hence by a fundamental physical theory that initially was thought to be irrelevant for studies of complexity. 
There are problems a quantum computer could efficiently solve that are believed to be intractable on any classical computer. 

In this way, it seems that the strong Church-Turing thesis has to be replaced by a quantum version \cite{Vazirani}. Colloquially speaking, the quantum Church-Turing thesis says that any process that can happen in nature that one could think of as being some sort of computation is efficiently simulatable:\\ 
\textbf{Strong quantum Church-Turing thesis.} % unfortunately \newtheorem* apparently does not work with revtex...
\textit{Every quantum mechanical computational process can be simulated efficiently in 
the unitary circuit model of quantum computation.}

Indeed, this notion of quantum computers being devices that can efficiently simulate natural 
quantum processes, being known under the name ``quantum simulation,'' is the topic of an entire 
research field initiated by the work of Feynman \cite{Feynman}. 
Steps towards a rigorous formulation have been taken by Lloyd \cite{Lloyd} and many others \cite{Others}.

Quite surprisingly, a very important class of physical processes appears to have been omitted in the quest for finding a sound theory of quantum simulation, namely \emph{dissipative quantum processes}.
Such processes are particularly relevant since, in the end, every physical process is to some extent  dissipative.
If one aims at simulating a quantum process occurring in a lab, 
one cannot, however, reasonably require the inclusion of all modes of the environment to which the 
system is coupled into the simulation.
Otherwise, one would always have to simulate all the modes of the environment, eventually of the entire universe, rendering the task of simulation obsolete and futile. 
We argue that the most general setting in which one can hope for efficient simulatability is the one of Markovian dynamics \cite{Markov} with arbitrary piecewise continuous time dependent control \cite{pedanticMartin}.
In any naturally occurring process the Liouvillian $\mcL$ determining the equation of motion 
\begin{align}\label{ME}
	\ddt \rho(t) = \mcL_t(\rho(t))
\end{align}
of the system state $\rho$ is \emph{$k$-local}. 
This means that the system is multipartite and $\mcL$ can be written as a sum of Liouvillians each acting nontrivially on at most $k$ subsystems. 
In fact, all natural interactions are two-local in this sense. 
Since we are interested in processes which can be viewed as a computation, we assume that the subsystems are of fixed finite dimension.
This is arguably the broadest class of natural physical processes that should be taken into account in a dissipative Church-Turing theorem and includes the Hamiltonian dynamics of closed systems as a special case.

In this work, we show the following.
 
(i) Every time evolution generated by a $k$-local time dependent Liouvillian can be simulated 
by a unitary quantum circuit with resources scaling polynomially 
in the system size $N$ and simulation time $\tau$.

(ii) As a corollary, we obtain that the dissipative model for quantum computing \cite{VerWolCir09}
can be reduced to the circuit model -- proving a conjecture that was still open. 

(iii) Technically, we show that the dynamics can be approximated by a Trotter decomposition, 
giving rise to a circuit of local channels, actually being reminiscent of the situation of unitary dynamics. 
In particular, in order to reach a final state that is only $\epsilon$ distinguishable from the exactly time evolved state, it will turn out to be sufficient to apply a circuit of $K m$ local quantum channels, where
\begin{align}
  \label{numbertsteps}
  m = \left\lceil \max\left( \frac{2cK^2 \tau^2}{\epsilon}, \frac{\tau b}{\ln 2} \right) \right\rceil 
\end{align}
is the number of time steps, 
$K\leq N^k$ is the number of local terms in the Liouvillian, and $b$ and $c$ are constants independent of $N$, $\tau$, $K$, and $\epsilon$. 
Some obstacles of naive attempts to simulate dissipative dynamics are highlighted,
and the specific role of the appropriate choice of norms is emphasized.

(iv) We also show that most quantum states cannot be prepared efficiently.

(v) In addition, the Trotter decomposition with our rigorous error bound is a practical 
tool for the numerical simulation of dissipative quantum dynamics on classical computers.

\paragraph{Setting.} %\paragraph should be used according to the APS style guide
We consider general quantum systems consisting of $N$ subsystems of Hilbert space dimension $d$.
The dynamics is described by a quantum 
master equation \eqref{ME} 
with a $k$-local Liouvillian of the form 
\begin{align} \label{LX}
  \mcL = \sum_{\Lambda\subset [N]} \mcL_\Lambda \ ,
\end{align}
where $[N]\coloneqq\{1,2,\ldots, N\}$ and $\mcL_\Lambda$ are \emph{strictly $k$-local} Liouvillians.
The subscript $\Lambda$ means that the respective operator or superoperator acts nontrivially only on the subsystem $\Lambda$ and we call an operator or superoperator \emph{strictly $k$-local} if it acts nontrivially only on at most $k$ subsystems. 
Each of the Liouvillians $\mcL_\Lambda$ can be written \cite{WolCir08} in Lindblad form \cite{Lin76}
\begin{align}\label{lindbladform}
\mcL_\Lambda=-\i [H_\Lambda, \argdot]+ \sum_{\mu=1}^{d^k} \mcD[L_{\Lambda,\mu}] \ ,
\end{align}
where $\mcD[X](\rho) \coloneqq 2 X \rho X\ad - \{ X\ad X , \rho \}$ and may depend on time piecewise 
continuously. 
In particular, we do not require any bound on the rate at which the Liouvillians may change.

The \emph{propagators} $\T[L]ts$ are the family of superoperators defined by
\begin{align}\label{defTeq}
\rho(t) = \T[L]{t}{s}(\rho(s))
\end{align}
for all $t\geq s$. 
They are completely positive and trace preserving (CPT) and uniquely solve the initial value problem
\begin{align}\label{odeT}
\ddt T(t,s)= \mcL_t T(t,s) \ , \quad T(s,s)=\id \ .
\end{align}

The main result, which is a bound on the error of the Trotter decomposition, will 
be somewhat reminiscent of the Trotter formula for time dependent Hamiltonian dynamics derived in Ref.\ \cite{HuyRae90}.
The main challenge comes from the fact that we are dealing with superoperators rather than operators.
The key to a meaningful Trotter decomposition is the choice of suitable norms for these superoperators. 
The physically motivated and strongest norm is the one arising from the operational distinguishability of 
two quantum states $\rho$ and $\sigma$, which is given by the trace distance 
$\dist(\rho,\sigma) \coloneqq \sup_{0 \leq A \leq \1} \tr(A(\rho-\sigma))$.
The trace distance coincides up to a factor of $1/2$ with the distance induced by the Schatten $1$-norm 
$\nnorm{\argdot}_1$, 
where the Schatten $p$-norm of a matrix $A$ is $\nnorm{A}_p\coloneqq(\tr(|A|^p))^{1/p}$.
Therefore, we measure errors of approximations of superoperators with the induced 
operator norm, which is the so-called $(1\rightarrow 1)$-norm.
In general the $(p \rightarrow q)$-norm of a superoperator $T\in \mc B( \mc B( \mc H))$ 
is defined as \cite{Wat05}
\begin{align}
  \norm{T}_{p\rightarrow q} \coloneqq \sup_{\nnorm{A}_p=1} \nnorm{T(A)}_q \ .
\end{align}
The difficulty in dealing with these norms lies in the fact that for $p< \infty$ the $p$-norm does not respect $k$-locality, e.g., $\lVert A\otimes \1_{n\times n} \rVert_1=n\lVert A \rVert_1$.
This problem is overcome by using the Lindblad form of the strictly $k$-local Liouvillians.
In the end, all bounds can be stated in terms of the largest operator norm $\norm{X_t}_\infty$ of the Lindblad operators 
$X\in \mcL_\Lambda$ of the strictly $k$-local terms. 
The notation $X\in \mcL_\Lambda$ means that $X$ is one of the operators occurring in the Lindblad representation \eqref{lindbladform} of $\mcL_\Lambda$.
From now on we assume that this largest operator norm $a$ is everywhere bounded by a constant of order $1$ and, in particular, independent of $N$, i.e., $a \in \landauO(1)$.
%
%%% ============ Thm. (Full Trotter) ==============

\paragraph{Main result.}
One can always approximate
any dissipative dynamics generated by a $k$-local Liouvillian acting on $N$ subsystems, even allowing for piecewise continuous time dependence, by a suitable Trotter decomposition.
The error made in such a decomposition can be bounded rigorously.

\begin{theorem}[\bf Trotter decomposition of Liouvillian dynamics]\label{fulltrotterthm}
Let $\mcL= \sum_{\Lambda \subset [N]} \mcL_\Lambda$ be a $k$-local Liouvillian that acts on $N$ subsystems with local Hilbert space dimension $d$. Furthermore, let the $\mcL_\Lambda$ be piecewise continuous in time 
with the property that $a=\max_\Lambda \max_{X\in \mcL_\Lambda} \sup_{t\geq 0} \norm{X_t}_\infty \in \landauO(1)$.
Then the error of the Trotter decomposition of a time evolution up to time $\tau$ into $m$ time steps is
\begin{equation}\label{fulltrotter}
	\BNorm{
	T_\mcL(\tau,0)- \prod_{j=1}^m\prod_{\Lambda \subset [N]}
	\T[L_{\mathrm{\Lambda}}]{\tau \tfrac{j}{m}}{\tau \tfrac{j-1}{m}}}
	\leq \frac{c K^2  \tau^2 \e^{b \tau/m}}{m},
\end{equation}
where $c\in \landauO(d^{2k})$, $b\in \landauO(d^k)$, and 
$K \leq N^k$ is the number of strictly $k$-local terms $\mcL_\Lambda \neq 0$. 
This bound holds for any order in which the products over $\Lambda$ are taken.
$\T[L_{\mathrm{\Lambda}}]{\tau \tfrac{j}{m}}{\tau \tfrac{j-1}{m}}$ can be replaced by the propagator 
$\T[L_{\mathrm{\Lambda}}^\av]{\tau \tfrac{j}{m}}{\tau \tfrac{j-1}{m}}
=\exp(\tau/m \mcL_\Lambda^\av)$ 
of the average Liouvillian
\begin{align}
\mcL_\Lambda^\av=
 \frac m \tau \int_{\tau (j-1)/m}^{\tau j/m} \mcL_\Lambda \rmd t
\end{align}
without changing the scaling \eqref{fulltrotter} of the error.
\end{theorem}

All constants are calculated explicitly in \theApndx.
The supremum in $a$ can be replaced by suitable time averages over the time steps such that 
$\norm{X_t}_\infty$ can be large for small times. Before we turn to the proof of this result, we discuss important implications.

\begin{implication}[\bf Dissipative Church-Turing theorem] 
$\ $\\
Time dependent Liouvillian dynamics can be simulated efficiently in the standard unitary circuit model.
\end{implication}
Using the Stinespring dilation \cite{Stinespring}, each of the $Km$ propagators 
$\T[L_{\mathrm{\Lambda}}]{\tau \tfrac{j}{m}}{\tau \tfrac{j-1}{m}}$ can be implemented as a unitary $U^j_\Lambda$ acting on the subsystem $\Lambda$ and an ancilla system of size at most $d^{2k}$.
These unitaries can be decomposed further into circuits $\tilde U^j_\Lambda$ of at most 
$n=\landauO(\log^\alpha(1/\epsilon_{\text{SK}}))$ gates from a suitable gate set 
using the Solovay-Kitaev algorithm \cite{DawNie06} with $\alpha <4$ 
such that 
$\nnorm{U^j_\Lambda-\tilde U^j_\Lambda}_\infty \leq \epsilon_{\text{SK}}$.
Note that for pure states, we have 
$\tkw 2 \nnorm{U \ketbra \psi \psi U\ad - \tilde U \ketbra \psi \psi \tilde U\ad}_1
\leq \nnorm{U - \tilde U}_\infty \leq \epsilon_{\text{SK}}$
and the $1$-norm is nonincreasing under partial trace.
The full error is bounded by the error from the Trotter approximation \eqref{fulltrotter} 
plus the one arising from the Solovay-Kitaev decomposition, in $(1\rightarrow 1)$-norm bounded by
$K m \epsilon_{\text{SK}}$.

At this point a remark on the appropriate degree of generality of the above result is in order. 
The proven result applies to dynamics under arbitrary piecewise continuous time dependent $k$-local Liouvillians.
It does not include non-Markovian dynamics as often resulting from strong couplings.
However, not only this result, but no dissipative Church-Turing theorem, can or should cover such a situation:
Including highly non-Markovian dynamics would mean to also include extreme cases such as an evolution implementing a swap gate that could write the result of an incredibly complicated process happening in the huge environment into the system. 
In such an intertwined situation it makes only limited sense to speak of the time evolution of the system alone in the first place.
On the other hand, in practical simulations of non-Markovian dynamics, where the influence of
memory effects is known, pseudo\-modes can be included \cite{Imamoglu}, 
thereby rendering the above results again applicable.

It has been shown recently \cite{PouQarSom11} that the set of states that can be reached from a fixed pure reference state 
by $k$-local, time dependent Hamiltonian dynamics is exponentially smaller than the set of all pure quantum states.
In fact, a more general statement holds true \seealsoepaps:

\begin{implication}[\bf Limitations of efficient state generation]\label{Imp:epsnets}
Let $X_\tau^\rho$ be the set of states resulting from the time evolution of an arbitrary initial state $\rho$ under all possible (time dependent) $k$-local Liouvillians up to some time $\tau$. For times $\tau$ that are polynomial in the system size, the relative volume of $X_\tau^\rho$ (measured in the operational metric induced by the $1$-norm) is exponentially small.
\end{implication}

Finally, Theorem~\ref{fulltrotterthm} also provides 
a rigorous error bound for the simulation of local time dependent Liouvillian dynamics on a classical computer. Even though classical simulation of quantum mechanical time evolution is generally believed to be hard in time, we have the following result.

\begin{implication}[\bf Simulation on classical computers \cite{LRTrotter}]\label{imp:CalssicalSim}
  For systems with short-range interactions and fixed time $\tau$, the evolution of 
  local observables can be simulated on classical 
  computers with a cost independent of the system-size and arbitrary precision, e.g., 
  using a variant of the \emph{time-dependent density matrix renormalization 
    group} method.
\end{implication}

To approximate the evolution of local observables in the Heisenberg picture, one applies the Hilbert-Schmidt adjoint, $\prod_{j=m}^1\prod_{\Lambda \subset [N]} T\ad_{\mcL_\Lambda}(\tau \tfrac{j}{m},\tau \tfrac{j-1}{m})$, of the Trotter approximation of the propagator to the observable.
Lieb-Robinson bounds can be used to prove that dissipative dynamics under short-range Liouvillians is \emph{quasi-local} and that the time evolution of local observables is restricted to a causal cone \cite{LRTrotter}.
Channels of the Trotter circuit that lie outside the causal cone have only a negligible effect and can be removed.
The error for the resulting approximation to the time evolved observable measured in the $\infty$-norm, is system-size independent.

This establishes a mathematically sound
foundation for simulation techniques based on Trotter decomposition that have 
previously been used without proving that the approximation is actually 
possible; see, e.g., Refs.~\cite{simulation}.

Recently, CPT maps like the local channels in the Trotter decomposition \eqref{fulltrotter} have even been implemented in the lab \cite{lab}.

\paragraph{Proof of theorem \ref{fulltrotterthm}.}
We now turn to the proof of the main result.
First we will find $(1 \rightarrow 1)$-norm estimates (i) for $T$ and (ii) for $T^-$ which will be used frequently. 
In the next step (iii) we derive a product formula, which we use iteratively (iv) to prove the Trotter decomposition. Finally, (v) we show how the second claim of the theorem concerning the approximation with the average Liouvillian can be proven.
Throughout the proof we consider times $t\geq s\geq 0$.

(i) 
Because any CPT map $T$ maps density matrices to density matrices, we have $\Norm{T}\geq 1$.
In Ref.\ \cite{Wat05} it is shown that 
\begin{align}
	\Norm{T}= \sup_{A=A\ad, \norm{A}_1=1} \norm{T(A)}_1 
\end{align}
for any CPT map $T$. Any self-adjoint operator
$A=A_+ - A_-$ can, by virtue of its spectral decomposition, be written as the difference of a positive and negative part $A_\pm\geq 0$. Since $T$ is CPT, 
$\norm{T(A_\pm)}_1=\tr(T(A_\pm))=\norm{A}_1$, hence $\Norm{T}\leq 1$, and finally 
$\Norm{T}=1$.

(ii)
For any Liouvillian $\mcK$ the propagator $\T[K]ts$ is invertible and the inverse $\iT[K] t s=(\T[K]ts)^{-1}$ is the unique solution of
\begin{align} \label{odeiT}
\ddt T^-(t,s) = -T^-(t,s) \mcK_t \ , \quad T^-(s,s)=\id \ .
\end{align}
From the representation of $T^-$ as a reversely time-ordered exponential, the inequality 
\begin{align}
	\Norm{\iT[\mcK]ts} \leq \exp (\int_s^t\Norm{\mcK_r} \rmd r ) 
\end{align}
follows.
This can be proved rigorously with the ideas from Ref.~\cite{DolFri77} \seealsoepaps.

For the case where $\mcK$ is strictly $k$-local, we use its Lindblad representation and the inequality 
$\Vert A \rho B\Vert_1 
	\leq \Vert A\Vert_\infty  \Vert \rho \Vert_1  \Vert B\Vert_\infty$
to establish $\Norm{\mcK} \in \landauO(d^k)$ and hence
 $ \Norm{\iT[\mcK]ts} \leq \e^{b(t-s)}$, 
with $b \in \landauO(d^k)$.

(iii) In the first step we use similar techniques as the ones being used for the unitary case \cite{HuyRae90} where differences of time evolution operators are bounded in operator norm by commutators of Hamiltonians.
Applying the fundamental theorem of calculus twice, one can obtain for any two Liouvillians $\mcK$ and~$\mcL$ 
\begin{widetext}
\begin{align}\nonumber
\T[K+L]ts -\T[K]ts \T[L]ts &=
\T[K]ts \T[L]ts
\int_s^t  \iT[L]rs \int_s^r  \frac{\rmd}{\rmd u} 
\left(\iT[K]us \mcL_r \T[K]us \right) \iT[K]rs \T[K+L]rs \, \rmd u \, \rmd r
\\ \label{longesteq}
&=\int_s^t  \int_s^r  \T[K]ts \T[L]tr \iT[K]us [\mcK_u, \mcL_r] \iT[K]ru \T[K+L]rs \, \rmd u\, \rmd r \ .
\end{align}
\end{widetext}
In the next step we take the $(1\to 1)$-norm of this equation, use the triangle inequality, 
employ submultiplicativity of the norm, and use (i) and (ii) to obtain
$\int_s^t  \int_s^r \Norm{ [\mcK_u, \mcL_r]  }\rmd u\, \rmd r$ as an upper bound.
In the case where $\mcK$ and $\mcL$ are strictly $k$-local 
$\Norm{[\mcK_u, \mcL_r]} \in \landauO(d^{2k})$, which follows by the same arguments used in (ii) to bound $\Norm{\mcK}$.
In the case where $\mcL$ is only $k$-local with $K$ terms, $\Norm{[\mcK_u, \mcL_r]}$ is increased by at most the factor $K$ such that
\begin{equation}
\Norm{\T[K+L]ts -\T[K]ts \T[L]ts} 
\in \landauO((t-s)^2 \e^{b(t-s)} d^{2k} K )  . \label{productformula}
\end{equation}

(iv) 
The propagator can be written as 
\begin{align} \label{Tdecomp}
T_\mcL(\tau,0) = \prod_{j=1}^m T_{\mcL}(  \tau j/m, \tau (j-1)/m) \ .
\end{align}
Using the inequality
\begin{align}\label{usualtriangleinequalitytrick}
\nnorm{T_1T_2 -\tilde T_1 \tilde T_2} \leq 
\nnorm{T_1}\nnorm{T_2 - \tilde T_2} + \nnorm{T_1 - \tilde T_1} \nnorm{\tilde T_2}
\end{align}
and Eq.\ \eqref{productformula} iteratively, one can establish the result as stated in Eq.~\eqref{fulltrotter}.

(v) 
For any strictly $k$-local Liouvillian $\mcK$ the propagator $\T[K]{t}{s}$ can be approximated by the propagator of the average Liouvillian,
\begin{equation} \label{avformula}
\nNorm{\T[K]ts-  \exp ( \int_s^t \mcK_r  \rmd r )}
 = \tfrac 13 b (t-s)^2 \ .
\end{equation}
This can be shown using the techniques described above by lifting the proof from Ref.~\cite{PouQarSom11} to the 
dissipative case \seealsoepaps.
A comparison of Eq.~\eqref{avformula} with Eq.~\eqref{productformula} shows that the error introduced by using the average Liouvillian is small compared to the error introduced by the product decomposition and does not change the scaling of the error.

\paragraph{Conclusion.}
In this work we show that under reasonable assumptions the dynamics of open quantum systems can be simulated efficiently by a circuit of local quantum channels in a Trotter-like decomposition. 
This channel circuit can further be simulated by a unitary quantum circuit with polynomially many gates from an arbitrary universal gate set. 
As a corollary it follows that the dissipative model of quantum computation is no more powerful than the standard unitary circuit model. 
The result can also be employed for simulations on classical computers and in the physically relevant case where the Liouvillian only has short-range interactions the simulation of local observables can be made efficient in the system size.
It also shows that systems considered in the context of dissipative phase transitions \cite{VerWolCir09,Zoller} 
can be simulated in both of the above senses. 
The result can be seen as a quantum Church-Turing theorem in the sense that under reasonable and necessary requirements any general time evolution of an open quantum system can be simulated efficiently on a quantum computer.

 \paragraph{Acknowledgements.} 
 This work was supported by the EU (Qessence, Minos, COQUIT, Compas), the BMBF (QuOReP), the 
 EURYI, the Niels Bohr International Academy, the German National Academic Foundation, and the Perimeter Institute. We would like to thank M.\ P.\ M\"{u}ller and T.\ Prosen for discussions.

%%% ============================================================
%%% ======================= APPENDIX ===========================
%%% ============================================================

%\clearpage
%\break
%\vspace{3\baselineskip}
\appendix

\section{\textbf{\textsc{Appendix}}}
%\section{\it \large Appendix}
In this appendix we elaborate on some of the technical aspects of our results and give explicit expressions for all involved constants.
First, we give a detailed derivation of the error caused by the Trotter approximation for the time evolution under a time dependent $k$-local Liouvillian.
Along the way we also derive a completely general bound for the Trotter error for arbitrary (not necessarily $k$-local) time dependent Liouvillians, which we don't need directly for the statements made in the paper, but which could be of interest independently as the bounds of our more specialized theorem is not optimal in certain situations.
Secondly, we present a detailed derivation of the error that is made when the time evolution under the time dependent Liouvillian is replaced by that of the average Liouvillian on a small time step.
Finally, we prove results on the scaling behavior of $\epsilon$-nets used in Implication~\ref{Imp:epsnets} to argue that only an exponentially small subset of states can be prepared with time dependent $k$-local Liouvillian dynamics in polynomial time from a fixed reference state. Our argument lifts the considerations from Ref.~\cite{PouQarSom11} to the space of density matrices and the physically relevant trace distance.

\section{Trotter approximation for time dependent Liouvillians}
We start by giving a detailed proof that for short time intervals it is possible to approximate the time evolution of a $k$-local time dependent Liouvillian $\mcK+\mcL$ by splitting off a strictly $k$-local part $\mcK$ and performing the time evolution under $\mcL$ and $\mcK$ sequentially.

\begin{theorem}[\bf Product decomposition of propagators] \label{prodthm}
Let $\mcL$ and $\mcK$ be two time dependent Liouvillians that act on the same quantum system of $N$ subsystems with local Hilbert space dimension $d$.
Furthermore, let $\mcK$ be strictly $k$-local and let $\mcL$ be $k$-local consisting of $K$ strictly $k$-local terms $\mcL_\Lambda$.
For $t\geq s$ the Trotter error is given by
\begin{align}
\Norm{\T[K+L]{t}{s}-\T[K]{t}{s} \T[L]{t}{s}} \leq (t-s)^2 \e^{b(t-s)}c\, K \ ,
\end{align}
where 
\begin{align*}
b&=\constb,
\\
c&=\constc,
\\
a&= \max_\Lambda \max_{X \in \mcK \cup \mcL_\Lambda} \sup_{s \leq v \leq t} \norm{X_v}_\infty .
\end{align*}
\end{theorem}
We will use this theorem iteratively to bound the error caused by decomposing the propagators of arbitrary $k$-local Liouvillians into the propagators of the individual strictly $k$-local terms.

The proof of this theorem can be presented most conveniently as a series of Lemmas.
From the main text (point~(i) in the proof of the Theorem, page~3) we already know that completely positive and trace preserving (CPT) maps are contractive:
\begin{lemma}[Contraction property of the propagator] \label{normlemma1} $ \ $ \\
  Let $T$ be a CPT map. Then $\Norm{T}=1$.
\end{lemma}
We also need to bound the norm of the inverse propagator.
\begin{lemma}[Backward time evolution] \label{normlemma2}
  For $t \geq s $:
\begin{enumerate}
\item[(i)]
$\T[L]ts$ 
is invertible and its inverse is $\iT[L]ts$ as defined by Eq.\ \odeiT in the main text.
\item[(ii)] If the Liouvillian $\mcL$ is piecewise continuous in time then
  \begin{equation}
    \Norm{\iT[L]{t}{s}} \leq \exp(\int_s^t \Norm{\mcL_r}\, \rmd r )    \ .
  \end{equation}
  \end{enumerate}
\end{lemma}
\begin{proof}
First, we consider the case where $\mcL$ is continuous in time and use the theory presented in Ref.~\cite{DolFri77} and in particular the ``properties'' which are proven in this reference.
The \emph{product integral} of $\mcL$ is defined analogously to the Riemann integral,
  \begin{align}
    \prod_s^t \exp(\mcL_r \, \rmd r) \coloneqq
    \lim_{\Delta r_j \rightarrow 0 \, \forall j} \prod_{j=1}^J \exp(\mcL_{r_j} \Delta r_j) \ ,
  \end{align}
  where $\prod_{j=1}^{J} X_j\coloneqq X_J X_{J-1}\ldots X_1$.
Since $\T[L]ts$ solves the initial value problem in Eq.~\odeT from the main text, $\T[L]{t}{s}= \prod_s^t \exp(\mcL_r \, \rmd r)$ which is exactly the statement of property 1. 

(i)
Property 3 precisely states that a product integral is invertible. It is not hard to see that the inverse of $\T[L] ts$ solves the initial value problem \odeiT from the main text.

(ii) 
The inverse propagator is
\begin{align}
\iT[L] t s = \left(\prod_s^t\exp(\mcL_r\, \rmd r)\right)^{-1} \ .
\end{align}
%%%%
Since matrix inversion is continuous, 
  \begin{equation}\label{iTasreversely}
  \iT[L] t s   = 
    \lim_{\Delta r_j \rightarrow 0 \, \forall j} \prod^{1}_{j=J} \exp(-\mcL_{r_j} \Delta r_j) .
  \end{equation}
We call this the \emph{reversely ordered product integral} and use the convention 
$\prod_{j=J}^{1} X_j\coloneqq X_1 X_2 \dots X_J$.
Using the submultiplicativity of the $(1 \to 1)$-norm and the triangle inequality we \nopagebreak obtain from Eq.~\eqref{iTasreversely} 
\begin{align} 
\Norm{\iT[L] ts} & \leq
\lim_{\Delta r_j \rightarrow 0 \, \forall j} \prod^{1}_{j=J}\exp(\Norm{\mcL_{r_j}} \Delta r_j)\\
= &
\exp(\lim_{\Delta r_j \rightarrow 0 \, \forall j} \sum_{j=1}^{J} \Norm{\mcL_{r_j}} \Delta r_j)
\end{align}
%\pagebreak
The definition of the Riemann integral finishes the proof for the continuous case.

If $\mcL$ is only piecewise continuous in time then (i) and (ii) hold for all the intervals where 
$\mcL$ is continuous and from that and the composition property $\T[L]uv \T[L] vw=\T[L]uw$ ($u\geq v \geq w$) it follows that (i) and (ii) hold on the whole time interval $[s,t]$.
\end{proof}
%\pagebreak
With these tools at hand we can now prove a bound on the Trotter error of two arbitrary (not necessarily $k$-local) time dependent Liouvillians.
\vspace{5\baselineskip}

\begin{widetext}
\begin{theorem}[\bf General Trotter error]\label{generaltrotter}
  For two arbitrary time dependent Liouvillians $\mcK$ and $\mcL$ the Trotter error is given by
  \begin{align}
    \Norm{\T[K+L]{t}{s}-\T[K]{t}{s} \T[L]{t}{s}} 
%\nonumber\\
     \label{eq:arbitraryLindbladianstrottererroraverage}
    \leq & \int_s^t\int_s^r
    \Norm{[\mcK_u,\mcL_r]}
    \,\rmd u\, \rmd r \ 
    \e^{2 \int_s^t\Norm{\mcK_v}\, \rmd v}      \\
    \label{eq:arbitraryLindbladianstrottererror}
    \leq &\tkw 2 (t-s)^2 \,
    \sup_{t\geq r \geq u \geq s} \Norm{[\mcK_u,\mcL_r]}\, 
    \exp\Bigl({2(t-s) \sup_{t\geq v \geq s} \Norm{\mcK_v} }\Bigr)      \ .
\end{align}
\end{theorem}
%\begin{widetext}
\begin{proof}
  We use a similar argument as in Ref.~\cite{HuyRae90}.
  With the fundamental theorem of calculus we obtain 
    \begin{align}\nonumber
      \iT[L]{t}{s} \iT[K]{t}{s} \T[K+L]{t}{s} - \id
      &= \int_s^t \del_r \bigl(
      \iT[L]{r}{s} \iT[K]{r}{s} \T[K+L]{r}{s} \bigr) \, \rmd r
      \\ \nonumber
      &= \int_s^t \iT[L]{r}{s} [\iT[K]{r}{s}, \mcL_r] \T[K+L]{r}{s} \, \rmd r
      \\ \nonumber
      &= \int_s^t \iT[L]{r}{s} \bigl(\iT[K]{r}{s} \mcL_r \T[K]{r}{s} - \mcL_r \bigr) \iT[K]{r}{s} \T[K+L]{r}{s} \, \rmd r
      \\ \nonumber
      &= \int_s^t \iT[L]{r}{s}
      \int_s^r \frac{\rmd}{\rmd u} \Bigl( \iT[K]{u}{s} \mcL_r \T[K]{u}{s} \Bigr)\rmd u\,
      \iT[K]{r}{s} \T[K+L]{r}{s}
      \, \rmd r
      \\ \nonumber
      &= \int_s^t \int_s^r \iT[L]{r}{s}
      \iT[K]{u}{s} [\mcL_r,\mcK_u] \T[K]{u}{s}
      \iT[K]{r}{s} \T[K+L]{r}{s}\, \rmd u\, \rmd r \ .
    \end{align} 
    Multiplying with $\T[K]{t}{s} \T[L]{t}{s}$ from the left yields 
    \begin{align}
      \T[K+L]{t}{s}- \T[K]{t}{s}\T[L]{t}{s}=\int_s^t\int_s^r
      \T[K]{t}{s} \T[L]{t}{r} \iT[K]{u}{s} [\mcL_r,\mcK_u] \iT[K]{r}{u} \T[K+L]{r}{s}
      \,\rmd u\, \rmd r.
    \end{align}
With submultiplicativity of the $(1\to 1)$-norm and the bounds on the norms of the forward and backward propagators from Lemma~\ref{normlemma1} and \ref{normlemma2} the result follows.
\end{proof}
\end{widetext}

To complete the proof of Theorem~\ref{prodthm} one needs to bound the norms $\Norm{[\mcL_r,\mcK_u]}$ and $\Norm{\mcK_r}$ in \eqref{eq:arbitraryLindbladianstrottererror} for the special case that $\mcK$ is strictly $k$-local and $\mcL$ is $k$-local with $K$ strictly $k$-local terms.%
%\\
%\pagebreak

\begin{lemma} \label{klocalandstrictlyklocalnormbounds}
Let $\mcK$ and $\mcL$ be two Liouvillians which act on the same operator space of $N$ subsystems with local Hilbert space dimension $d$. Furthermore, let $\mcK$ be strictly $k$-local and $\mcL$ be $k$-local consisting of $K$ strictly $k$-local terms $\mcL_\Lambda$.
Then
  \begin{align}
    \label{eq:klocalandstrictlyklocalnormbounds1}
    &           & 2 \Norm{\mcK_v}        &\leq b_v \\
    \label{eq:klocalandstrictlyklocalnormbounds2}
    &\text{and} & \tfrac{1}{2} \Norm{[\mcL_r,\mcK_u]} &\leq c_{r,u} K  \ ,
  \end{align}
  where $b_v=4a_v + 8d^k a_v^2$, $c_{r,u}=2a_ra_u+4(a_ra_u^2+a_r^2a_u) d^k +16 a_r^2a_u^2 d^{2k}$, and 
  $a_t= \max_\Lambda\max\{\norm{X_t}_\infty : X \in \mc K \cup \mcL_\Lambda \}$.
\end{lemma}

\begin{proof}
  First, let both Liouvillians be strictly $k$-local.
  Hence each of them can be written with at most $d^k$ Lindblad operators.
  Let the Lindblad representations of $\mcK$ and $\mcL$ be
  \begin{align}
    \label{eq:sumrepresentationofmcL}
     \mc K &= -\i [G, \argdot] + \sum_{\nu=1}^{d^k} \mcD[K_\nu] \\
    \intertext{and}
     \mcL &= -\i [H, \argdot] + \sum_{\mu=1}^{d^k} \mcD[L_\mu]  \ ,
  \end{align} 
  where $\mcD[X](\rho) \coloneqq 2 X \rho X\ad - \{ X\ad X , \rho \}$.
  Inequality \eqref{eq:klocalandstrictlyklocalnormbounds1} follows from counting the number of terms in 
  \eqref{eq:sumrepresentationofmcL} and using that 
  $\norm{A \rho B}_1 \leq \norm{A}_\infty \norm{\rho}_1 \norm{B}_\infty$.
  Similarly, by writing out the commutator $[\mcK, \mcL]$ and using the above representations one can verify that $[\mcK_r, \mcL_u] \leq 2a_ra_u+4(a_ra_u^2+a_r^2a_u) d^k +16 a_r^2a_u^2 d^{2k}$.
  If $\mcL = \sum_{\Lambda\subset [N]} \mcL_\Lambda$ is $k$-local with $K$ terms the bound is increased by at most a factor of $K$.
\end{proof}

Theorem~\ref{prodthm} follows as a corollary of Theorem~\ref{generaltrotter} and Lemma~\ref{klocalandstrictlyklocalnormbounds} by inserting the suprema of the bounds~\eqref{eq:klocalandstrictlyklocalnormbounds1} and~\eqref{eq:klocalandstrictlyklocalnormbounds2} into Eq.~\eqref{eq:arbitraryLindbladianstrottererror}.
Instead of using suprema in the step from Eq.~\eqref{eq:arbitraryLindbladianstrottererroraverage} to Eq.~\eqref{eq:arbitraryLindbladianstrottererror} one can take averages over $b_v$ and $c_{r,u}$ to obtain a better, but more complicated bound.
One can also improve the scaling of the error with the size of the time steps by using higher order Trotter schemes as in Ref.~\cite{HuyRae90} (time dependent case) or Ref.~\cite{suzuki} (time constant case).

\section{Approximation by the average Liouvillian}
In the product formula in our Theorem~\fulltrotterthm in the main text 
one can replace the time ordered integrals $\T[L_\mathrm{\Lambda}]ts$ by ordinary exponentials of the time averaged Liouvillians. This is not essential to our argument concerning the quantum Church-Turing thesis, but makes the result more useful for applications.
The additional error caused by doing this is bounded in the following theorem:
\begin{theorem}[\bf Approximation by the average Liouvillian] \label{avlem}
Let $\mcK$ be a strictly $k$-local Liouvillian acting on an operator space with local Hilbert space dimension $d$.
Then for any $t\geq s$
\begin{align}
\Norm{T_\mcK(t,s) - \exp((t-s)\mcK^\av)} 
& \leq\tfrac{1}{3} b (t-s)^2 \ ,
\end{align}
where the average Liouvillian
\begin{align}
\mcK^\av\coloneqq \kw{t-s} \int_s^t \mcK_r \, \rmd r
\end{align}
is indeed a Liouvillian, $b=\constb$, and 
$a= \max_{X_t \in \mcK}\sup_{t}  \norm{X_t}_\infty$.
\end{theorem}
\begin{proof}
We lift the proof from Ref.~\cite{PouQarSom11} to the dissipative setting.
Let $t\geq s $ be fixed. Applying the fundamental theorem of calculus and the definition of $\mcK^\av$, we obtain
\begin{widetext}
\begin{align*}
\T[\mcK^\av]ts-T_\mcK(t,s)&= -	\T[K]ts \int_s^t T^-_\mcK(u,s) \left(\mcK_u - \mcK^\av \right) T_{\mcK^\av}(u,s) \, \rmd u \\
&= - \kw{t-s} \int_s^t \int_s^t \T[K]tu \left(\mcK_u - \mcK_r \right) T_{\mcK^\av}(u,s) \, \rmd r\, \rmd u\\
&=  - \kw{t-s} \int_s^t \int_s^t \Bigl(
 \T[K]tu \mcK_u T_{\mcK^\av}(u,s)
-\T[K]tr \mcK_u T_{\mcK^\av}(r,s) 
\Bigr) \rmd r\, \rmd u \ .
\end{align*}
The inequality in Eq.\ \usualtriangleinequalitytrick from the main text yields
\begin{equation}
  \begin{split}
    \label{TTap}
\Norm{\T[\mcK^\av]ts-T_\mcK(t,s)} 
\leq 
\frac{1}{t-s} \int_s^t\int_s^t 
& \Norm{\mcK_u}\Bigl( \Norm{\T[K]tu-\T[K]tr} \Norm{T_{\mcK^\av}(u,s)} \\ 
&+ \Norm{\T[K]tr} \Norm{T_{\mcK^\av}(u,s)-T_{\mcK^\av}(r,s)} \Bigr) \rmd r\, \rmd u \ .    
  \end{split}
\end{equation}
\end{widetext}
From $\T[K]{u}{s} - \T[K]{r}{s} = -\int_r^u \T[K]{v}{s} \mcK_v \, \rmd v$, Lemma~\ref{normlemma1}, and the submultiplicativity of the norm we know that for $t\geq \nolinebreak u,r \geq s $
\begin{equation}\label{invav}
 \Norm{\T[K]us-\T[K]rs}  \leq  \left| \int_u^r \Norm{\mcK_v} \rmd v \right|  
\end{equation}
and similarly for $\mcK^\av$. With \eqref{TTap} we obtain 
\begin{multline}
\Norm{T_{\mcK^\av}(t,s)-T_\mcK(t,s)} \\
    \leq\,  2 \int_s^t \int_s^t \left| \int_u^r \Norm{\mcK_v} \rmd v \right| \rmd r \, \rmd u \ . 
\end{multline}
It remains to show that $\mcK^\av$ is a Liouvillian, i.e., that $\exp(t\mcK^\av)$ is a CPT map for all $t\geq 0$. First of all, finite sums of Liouvillians are Liouvillians. Furthermore, 
limits of sequences of Liouvillians are Liouvillians since the exponential function is continuous and the set of CPT maps is closed.
\end{proof}

\section{Efficiently preparable states constitute an exponentially small subset of state space}
In the following we will argue that for every fixed initial state, the time evolution for a time interval of length $\tau$ under any (possibly time dependent) $k$-local Liouvillian yields a state that lies inside of one of $N_T$ $\epsilon$-balls in trace distance. 
For times $\tau$ which are polynomial in the system size, $N_T$ is exponentially smaller than the cardinality of any $\epsilon$-net (in trace distance) that covers the state space $\mc S$. The case of Hamiltonian dynamics and state vectors is investigated in Ref.~\cite{PouQarSom11}.
It will be convenient to use the Bachmann-Landau symbols $\landauO$ and $\Omega$ for 
asymptotic upper and lower bounds up to constant factors.

By using Theorem~\ref{fulltrotterthm} of the main text, which provides an error bound for the Trotter approximation of a Liouvillian time evolution, together with the Stinespring dilation \cite{Stinespring} and the Solovay-Kitaev algorithm \cite{DawNie06}, one obtains the following:

\begin{theorem}[\bf Number of channel circuits]\label{thm:numbercircuits}
The propagator from time $0$ to time $\tau$, generated by any $k$-local time dependent Liouvillian acting on $N$ subsystems with local Hilbert space dimension $d \in \landauO(1)$ can be approximated in $(1\to 1)$-norm to accuracy $\epsilon>0$ with one out of $N_T$ channel circuits, where
\begin{align}
\log(N_T) \in \landauO\left(\frac{N^{3k+2} \tau^4}{\epsilon^5} \right) \ .
\end{align}
 \end{theorem}
 
\begin{proof}
According to Theorem~\ref{fulltrotterthm} of the main text, the propagator $T_\mcL(\tau,0)$ of the Liouvillian time evolution can be approximated by a circuit 
 $\prod_{j=1}^m\prod_{\Lambda \subset [N]}T_\Lambda^j$
of at most $N^k m$ strictly $k$-local channels $T_\Lambda^j$ to precision $\epsilon_1$ in $(1\to 1)$-norm, where according to Eq.~\numbertsteps from the main text, 
$m = 2cN^{2k} \tau^2 / \epsilon_1$. 
We have assumed that 
$2 \ln(2) c N^{2k} \tau/\epsilon_1 \geq b$ where $c$ and $b$ are given explicitly in Theorem~\ref{prodthm} and depend only on strictly local properties of the Liouvillian. 
Employing the Stinespring dilation \cite{Stinespring} for each of the channels $T_\Lambda^j$ one obtains a circuit of at most $N^k m$ strictly $3k$-local unitary gates $U_\Lambda^j$. Each $U_\Lambda^j$ acts on an enlarged system composed of the $d^k$-dimensional original subsystem and an ancilla system of dimension $d^{2k}$.
One can use the Solovay-Kitaev algorithm \cite{DawNie06} to approximate every single gate $U_\Lambda^j$ of the unitary circuit by a circuit $\tilde U_\Lambda^j$ of one- and two-qubit gates from a universal gate set of cardinality $n_{\SK} \in \landauO(1)$, e.g., $n_\SK=3$. 
With $N_\SK=c_{\SK}\log^\alpha(1/\epsilon_{\SK})$ of those $n_{\SK}$ standard gates, each unitary $U_\Lambda^j$ can be approximated to accuracy $\epsilon_{\SK}$ introducing a total error $\epsilon_2=N^k m \epsilon_\SK$. 
The constant $c_\SK$ depends on~$d^{3k}$.

Consequently, we have for the dilation $U$ of $\prod_{j=1}^m\prod_{\Lambda \subset [N]}T_\Lambda^j$ an approximation $\tilde U$ with operator norm accuracy $\epsilon_2$, given by a unitary circuit of $N_\text{All gates}=N_\SK N^k m$ standard gates from the universal gate set. 
Note that for any pure state $\ket \psi$, we have 
$\tkw 2 \nnorm{U \ketbra \psi \psi U\ad,\tilde U \ketbra \psi \psi \tilde U\ad}_1
\leq \nnorm{U - \tilde U}_\infty$
and the $1$-norm is non-increasing under partial trace. 
Tracing out the ancillas, we obtain an approximation $\tilde T$ of $\T[L]\tau 0$ with error
$\nNorm{\T[L]\tau 0-\tilde T} \leq \epsilon=\epsilon_1+2\epsilon_2$.
The total number of different channels $\tilde T$, which can arise in this way from the chosen universal gate set, is $N_T\leq {n_\SK}^{N_\text{All gates}}$, i.e., for given $c,\tau,k,N$ and $d$, a number of $N_T$ standard gates are enough to approximate any $\T[L]\tau 0$ in $(1\to 1)$-norm to accuracy $\epsilon$.

To conclude, we bound the order of $N_T$.
\begin{align}\nonumber
\log(N_T) &\leq N_\text{All gates} \log n_\SK\\ \nonumber
&= c_\SK \log^\alpha\left( \frac{2cN^{3k} \tau^2}{\epsilon_1\epsilon_2} \right) \frac{2cN^{3k} \tau^2}{\epsilon_1} \log n_\SK \\
&< c_\SK (3k)^\alpha \log^\alpha\left( \frac{2cN \tau}{\epsilon_1\epsilon_2} \right) \frac{2cN^{3k} \tau^2}{\epsilon_1} \log n_\SK \ .
\end{align}
Since we are interested in the scaling of $\log(N_T)$ for large $N$ and small $\epsilon_1, \epsilon_2$ we can assume that the argument of the logarithm is larger than $18$ and use that $\log_2^4(x)< x^2$ for $x\geq 18$ to obtain
\begin{align}
\log(N_T) &< 
C \frac{N^{3k+2} \tau^4}{\epsilon_1^3 \epsilon_2^2} 
\end{align}
with $C= c_\SK (3k)^\alpha (2c)^3 \log n_\SK$.
\end{proof}

The above theorem shows that the time evolution under a $k$-local Liouvillian can be approximated by one out of $N_T$ many circuits to accuracy $\epsilon$. 
The states that can be reached by any $k$-local Liouvillian time evolution, starting from a fixed initial state, are hence all contained in the union of $N_T$ $\epsilon$-balls (in $1$-norm) around the output states of these circuits.

Let us now determine whether those $\epsilon$-balls can possibly cover the whole state space.
For this purpose we introduce $\epsilon$-nets. 
We consider a $D$-dimensional Hilbert space $\mc H$ and denote 
\begin{compactenum}[(i)]
\item the set of state vectors, i.e., the set of normalized vectors in $\mc H$ by $P\subset \mc H$,
\item the set of density matrices by $\mc S \subset \mc B(\mc H)$, and
\item the set of rank one projectors by $\mc P \subset \mc S$.
\end{compactenum}
For an arbitrary subset $R\subset \mc B(\mc H)$ and some $\epsilon>0$ we call a finite subset $\mc N_\epsilon^p(R) \subset R$ satisfying 
\begin{align}
  \forall a \in R\ \exists b \in \mc N_\epsilon^p(R) : \norm{a-b}_p\leq \epsilon
\end{align}
\emph{an $\epsilon$-net for $R$ in (Schatten) $p$-norm}.
Furthermore, we call an $\epsilon$-net $\hat{\mc{N}}^p_\epsilon(R)$ \emph{optimal} if any other set $X\subset R$ with smaller cardinality $|X|<|\hat{\mc{N}}^p_\epsilon(R)|$ cannot be an $\epsilon$-net for $R$ in $p$-norm.
Similarly, we define $\epsilon$-nets $\mc N_\epsilon^\HS(P)\subset P$ for state vectors in Hilbert space norm and, as before, we denote optimal $\epsilon$-nets by $\hat{\mc N}_\epsilon^\HS(P)$.

In Ref.~\cite{HayLeuWin06} it was shown that for the set of state vectors of a $D$-dimensional quantum system there exist 
$\epsilon$-nets of cardinality at most $|\mc N_\epsilon^\mathrm{HS}(P)| \leq (5/(2 \epsilon))^{2D}$.
As the Hilbert space distance upper bounds \cite{HayLeuShoWin04} the trace distance, 
\begin{align}
|\ket \psi - \ket \phi|_2 & \geq \tkw 2 \norm{\ketbra \psi \psi-\ketbra \phi \phi}_1 \nonumber \\
& = \dist(\ketbra \psi \psi,\ketbra \phi \phi) \ ,
\end{align}
this also implies the existence of $\epsilon$-nets for $\mc P$ in $p$-norm of cardinality $|\mc N_\epsilon^p(\mc P)| \leq (5/\epsilon)^{2D}$ for any $p\geq1$.
By comparing the volume of the $\epsilon$-balls with the volume of the whole set of state vectors one can see that for state vectors this construction is essentially optimal.

\begin{lemma}\label{lem:epsnetscalingP}
 For a $D$-dimensional quantum system 
 \begin{equation}
  |\mc N_\epsilon^\mathrm{HS}(P)| \in  \Omega(\left(\kw \epsilon\right)^{2D-1}) \cap \landauO(\left(\frac{5}{2 \epsilon}\right)^{2D}) \ . 
 \end{equation}
\end{lemma}
\begin{proof}
The set of state vectors in a $D$-dimensional Hilbert space is isomorphic to a $(2D-1)$-sphere with radius $1$ in $(2D)$-dimensional real Euclidean space such that the Hilbert space norm $|\argdot|_2$ on state vectors coincides with the Euclidean norm in $\mathbb{R}^{2D}$.
The surface area of a $(n-1)$-sphere of radius $r$ is $S_{n-1}(r) = n C_{n} r^{n-1}$, where $C_n = \pi^{n/2}/\Gamma(n/2+1)$ and $\Gamma$ is the Euler gamma function. 
The set of states within Hilbert space distance $\epsilon$ to a given state is a spherical cap on that sphere with opening angle $4 \arcsin(\epsilon/2)$. For $\epsilon \ll 1$, the area of such a cap is approximately equal to the volume of a $(2D-1)$-ball of radius $\epsilon$.
In fact, a more detailed analysis reveals that for $D=3$ the two are exactly identical and for $D>3$ the cap is always smaller than the $(2D-1)$-ball.
The volume of an $n$-ball of radius $r$ is $V_n(r) = C_n r^n$.
Thus for $D\geq3$,
\begin{align*} %\label{eq:hilbertspacescaling}
  \left(\frac{5}{2 \epsilon}\right)^{2D} 
  &\geq |\mathcal{N}_\epsilon^{\mathrm{HS}}(P)| \geq \frac{S_{2D-1}(1)}{V_{2D-1}(\epsilon)} 
  = \frac{2 D C_{2D}}{C_{2D-1} \epsilon^{2D-1}}
\\
  &= 2\sqrt{\pi} \frac{\Gamma(D+1/2)}{ \Gamma(D)} \left(\frac{1}{\epsilon}\right)^{2D-1} \geq \frac{15 \pi}{8} \left(\frac{1}{\epsilon}\right)^{2D-1} , \nonumber
\end{align*}
where the first inequality follows from Ref.~\cite{HayLeuWin06}.
\end{proof}

This is essentially the argument used in Ref.\ \cite{PouQarSom11} to establish that Hilbert space is a ``convenient illusion".
However, the lower bound on $|\mc{\hat{N}}_\epsilon^{\mathrm{HS}}(P)|$ does not immediately imply a lower bound on $| \mc{\hat{N}}_\epsilon^p(\mc P)|$ (and hence also not for $| \mc{\hat{N}}_\epsilon^p(\mc S)|$) for any $p\geq1$.
In particular, there are states with distance $2$ in Hilbert space norm and distance $0$ in any of the $p$-norms,  namely, any pair of state vectors $\{\ket{\psi}, -\ket{\psi}\}$.

We now show that a similar lower bound as in the last lemma holds for the size of optimal $\epsilon$-nets for $\mc P$ and $\mc S$ in $p$-norm.

\begin{lemma}
  \label{epsilonnets}
  For $p \in \{1,2\}$
  \begin{equation}
    \label{eq:epsilonnets}
   |\mc{\hat{N}}_\epsilon^p(\mc S)| \geq 
   | \mc{\hat{N}}_{2\epsilon}^p(\mc P)| \in \Omega(\left(\frac{1}{4\epsilon}\right)^{2D-3}) .
  \end{equation}
\end{lemma}

\begin{proof} 
For a given state vector $\ket \psi$ it will be convenient to use the notation $\psi \coloneqq \ketbra \psi \psi$.

We start to prove the first inequality. Fix $p\in \{1,2\}$.
There is a family $\{\rho_j\} \subset \mc{\hat{N}}_\epsilon^p(\mc S)$ such that 
their $\epsilon$-neighborhoods in $p$-norm cover $\mc P$ and such that 
for each $\rho_j$ there exists a rank-1 projector $\psi_j \in \mc P$ satisfying $\norm{\rho_j - \psi_j}_p \leq \epsilon$. 
Then $\{\psi_j\}$ is a $(2\epsilon)$-net for $\mc P$ in $p$-norm with 
$|\mc{\hat{N}}_\epsilon^p(\mc S)| 
\geq |\{\psi_j\}| \geq  | \mc{\hat{N}}_{2\epsilon}^p(\mc P)|$.

From $\norm \argdot_1 \geq \norm \argdot _2$ it follows that $|\mc N^2_\epsilon(\mc P)| \leq |\mc N^1_\epsilon(\mc P)|$.
Hence it remains to prove the lower bound for $|\mc N^2_{2\epsilon}(\mc P)|$ in \eqref{eq:epsilonnets}.
For this we construct an $\epsilon'$-net $\mc N^\HS_{\epsilon'}(P)$ for state vectors in Hilbert space norm from a $(2\epsilon)$-net $\mc N^2_{2\epsilon}(\mc P)$.
For every element $\psi_j \in \mc N^2_{2\epsilon}(\mc P)$ we fix an eingenvalue-1 eigenvector $\ket{\psi_j}$.
Using the $(\epsilon^2/2)$-net 
$\mc N_{\epsilon^2/2}^1([0,1[)= \{\epsilon^2, 2 \epsilon^2, \ldots, \lceil 1/\epsilon^2 \rceil \epsilon^2 \}$ for $[0,1[$ with cyclic boundary conditions we define the set
\begin{align} \label{eq:epsprimenet}
\mc N_{\epsilon'}^\HS(P) =
\{ \e^{2\pi \i \delta} \ket{\psi} : \delta \in \mc N_{\epsilon^2/2}^1([0,1[), \ \ket{\psi} \in \{\ket{\psi_j}\} \} \ .
\end{align}
This is an $\epsilon'$-net for $P$ and we will find an expression for $\epsilon'$ in terms of $\epsilon$. 

Let $\ket \phi \in P$. Then there exists a state vector $\ket \psi \in \{\ket{\psi_j}\}$ such that 
\begin{align*}
(2\epsilon)^2 
    &\geq \norm{\phi - \psi}_2^2 = 2 - 2 |\braket{\phi}{\psi}|^2 \\
    &\geq  2 - 2 |\braket{\phi}{\psi}| \ ,
\end{align*}
and a $\delta \in \mc N_{\epsilon^2/2}^1([0,1[)$ such that
\begin{align*}
\left|\, \left|\braket \phi \psi \right|- \Re(e^{2\pi\i \delta }\braket \phi \psi) \right| < (2\epsilon)^2 \ .
\end{align*}
Together this yields
\begin{align}\nonumber
  3 (2 \epsilon)^2 > 2 - 2 \Re( \e^{2\pi\i \delta} \braket{\phi}{\psi}) =
| \ket\phi - \e^{2\pi\i \delta}\ket \psi |_2^2   \ .
\end{align}
Since $\e^{2\pi \i \delta} \ket \psi \in \mc N^\HS_{\epsilon'}(P)$, we can choose $\epsilon'=4\epsilon > \sqrt{12} \epsilon$ to make $\ {\mc N}_{\epsilon'}^\HS(P)$ a $(4\epsilon)$-net.
From the definition~\eqref{eq:epsprimenet} of $\mc N_{\epsilon'}^\HS(P)$ we can bound its cardinality
\begin{align}\nonumber
|\mc N_{4\epsilon}^\HS(P)|&=|\mc N_{\epsilon^2}^1([0,1[)|\, |\{\ket{\psi_i}\}|\\
\label{eq:epsilonnetinequality}
&< \lceil 1/\epsilon^2 \rceil |\mc N^2_{2\epsilon}(\mc P)| \ ,
\end{align}
where we have used that by construction $|\{\ket{\psi_i}\}| = |\mc N^2_{2\epsilon}(\mc P)|$. 
Finally, as the described construction works for any $(2\epsilon)$-net $\mc N^2_{2\epsilon}(\mc P)$, we obtain
\begin{align}
\left\lceil1/\epsilon^2\right\rceil \,|\hat{\mc N}^2_{2\epsilon}(\mc P)|
>
|\hat{\mc N}^{\mathrm{HS}}_{4\epsilon}(P)|
\end{align} 
and Lemma~\ref{lem:epsnetscalingP}  finishes the proof.
\end{proof}

Combining Theorem~\ref{thm:numbercircuits} and Lemma~\ref{epsilonnets}, we arrive at the following theorem:
\begin{theorem}[\bf Limitations of efficient state generation]\label{thm:onenormepsnet}
For every fixed initial state, the time evolution for a time interval of length $\tau$ under any $k$-local Liouvillian acting on $N$ subsystems with local Hilbert space dimension $d$ yields 
a state that lies inside one of $N_T$ $\epsilon$-balls in $1$-norm with
$\log(N_T) \in \landauO\left({N^{3k+2} \tau^4/\epsilon^5} \right)$.  
For times $\tau$ polynomial in the system size $N$, this is asymptotically exponentially smaller than 
$\log |\mc{\hat{N}}_\epsilon^1(\mc S)| \in \Omega(d^N/\log(1/\epsilon))$ where $|\mc{\hat{N}}_\epsilon^1(\mc S)|$ is the cardinality of an optimal $\epsilon$-net in $1$-norm that covers the state space~$\mc S$.
\end{theorem}


\begin{thebibliography}{99}
\bibitem{pedanticMartin}
	For finite dimensional systems this means mathematically that $T(t,s)$ defined in Eq.~\eqref{defTeq} satisfies 
	(i)~$T(t,r)T(r,s)=T(t,s)\ \forall t\geq r \geq s \geq 0$ and 
	(ii)~$\lim_{\epsilon \to 0} \norm{T(t+\epsilon,t) - \id}=0 \ \forall t\geq 0$. 
	Then $T$ is generated by a time dependent Liouvillian $\mcL_t$ \cite{WolCir08}.
\end{thebibliography}
\end{document}